\documentclass[11pt,hidelinks]{article}

\usepackage{amssymb, amsmath, amsthm, amsfonts, enumitem, latexsym, tikz, color, hyperref, soul}
\usepackage[margin=1in,paperwidth=8.5in,paperheight=11in]{geometry} 
\usepackage[nameinlink]{cleveref}
\usetikzlibrary{calc, shapes.multipart, arrows.meta, positioning} 
\tikzstyle{pt}=[circle, fill=black, inner sep=2pt]

\newcommand{\bruno}[1]{\color{blue!60} #1 \color{black}} 
\newcommand{\marius}[1]{\color{red!60} #1 \color{black}} 
\newcommand{\bcom}[1]{\bruno{[#1]}} \renewcommand{\st}[1]{\!\!} \renewcommand{\bcom}[1]{\!\!} 

\newtheorem{theorem}{Theorem}[section] \newtheorem{corollary}[theorem]{Corollary} \newtheorem{proposition}[theorem]{Proposition} 
\newtheorem{claim}[theorem]{Claim} \newtheorem*{claim*}{Claim}\newtheorem{lemma}[theorem]{Lemma} \newtheorem*{theorem*}{Theorem}
\theoremstyle{definition} \newtheorem*{definition}{Definition}
\theoremstyle{remark} \newtheorem*{remark*}{Remark} 

\newcommand{\eps}{\varepsilon}
\DeclareMathOperator{\prob}{Prob} 
\DeclareMathOperator{\poly}{poly}  
\DeclareMathOperator{\C}{C} \DeclareMathOperator{\D}{D} \DeclareMathOperator{\E}{E} \DeclareMathOperator{\F}{F} 
\newcommand{\mcX}{\mathcal X} \newcommand{\mcY}{\mathcal Y} \newcommand{\mcR}{\mathcal R} \newcommand{\mctR}{{\widetilde{\mathcal R}}}

\date{}
\title{Shorter lists with short programs via unbalanced $1$-expanders with almost optimal right size}
\author{
  Bruno Bauwens\thanks{
    National Research University Higher School of Economics, Faculty of Computer Science, Moscow, Russia.
    This article is an output of a research project (HSE-BR-2025-024) implemented as part of the Basic Research Program at HSE University. 
  }
  \and
  Marius Zimand\thanks{
    Department of Computer and Information Sciences, Towson University, Baltimore, MD.
  }
} 

\begin{document}
\maketitle

\if01
\begin{abstract}
  An explicit graph is given with left size $N$, left degree $\widetilde O(\log^2 N)$, right size $(1+o(1))K$ and $1$-expansion up to~$K$, 
  meaning that every left subset of size $K' \le K$ has at least $K'$ neighbors.

  Let $\C(x)$ be the minimal length of a program that prints~$x$ (i.e., the central concept in Kolmogorov complexity).  
  The $1$-expander is used to obtain an algorithm that on input $x$ computes in time $\poly(|x|)$ a list with $\widetilde O(|x|^3)$ programs 
  such that at least 1 program prints $x$ and has length $\C(x) + O(1)$. 
  The list size improves on the $O(|x|^{6+\eps})$  in~\cite{zim:c:shortlistshortproof} 
  and is close to the $\Omega(|x|^2)$ lower bound from~\cite[theorem 4]{bmvz:j:shortlist}. 

  In the companion paper~\cite{companion-DS}, the $1$-expander is used to obtain dynamic dictionaries in which the query operation has non-adaptive memory access. 
  
\end{abstract}
\fi

\begin{abstract}
  Let $\C(x)$ be the minimal length of a program that prints~$x$ (i.e., the central concept in Kolmogorov complexity).  
  We obtain an algorithm that on input $x$ computes in time $\poly(|x|)$ a list with $\widetilde O(|x|^3)$ programs 
  such that at least 1 program prints $x$ and has length $\C(x) + O(1)$. 
  This list size improves on the $O(|x|^{6+\eps})$ in~\cite{zim:c:shortlistshortproof} 
  and is close to the $\Omega(|x|^2)$ lower bound from~\cite[theorem 4]{bmvz:j:shortlist}. 

The proof relies on an explicit  construction of an unbalanced $1$-expander with almost optimal right size and small left degree. This is a bipartite graph with left size $N$, left degree $\widetilde O(\log^2 N)$, right size $(1+o(1))K$ and $1$-expansion up to~$K$, 
  meaning that every left subset of size $K' \le K$ has at least $K'$ neighbors. 

  In the companion paper~\cite{companion-DS}, the $1$-expander is used to obtain dynamic dictionaries in which the query operation has non-adaptive memory access. 
  
\end{abstract}

\section{Introduction}

This paper makes 2 contributions. The first one concerns  bipartite expanders in which the left side is much larger than the right side, hence, the name ``unbalanced.'' 
\begin{definition}
  A graph has $e$-expansion up to~$K$ if every left set $S$ of size at most~$K$ has at least $e|S|$ neighbors.
\end{definition}

A lot of attention has been given to lossless expanders, which are graphs with $(1-\eps)D$-expansion. 
In contrast, our first contribution is an explicit $1$-expander that has small left degree and small right side.
Such expanders are used in applications that require matching or some online variant of it.
One such application is in dynamic dictionary datastructures with non-adaptive probing, as proposed in~\cite{lar-pag-etc:c:nonadaptivecellprobe}.
Our graph mentioned in  the abstract  is used in the companion paper~\cite{companion-DS}, dedicated to this topic.
Another application is in our second contribution: the polynomial time generation of a list containing $O(1)$-shortest programs, which improves previous results from~\cite{bmvz:j:shortlist,teu:j:shortlists,zim:c:shortlistshortproof}. We  discuss it in the next section.
\smallskip

\noindent
Obviously, graphs with $1$-expansion up to~$K$ have right size at least~$K$.
For graphs with left size $N$ and right size $M$,
the lower bound on the degree is $\Omega((\log \tfrac N K )/(\log \tfrac M K))$,
see \cite[lemma 4.1 p50]{bmvz:j:shortlist}.

The table below gives existing constructions and our new one. 
\footnote{We use the notation $\smash{\widetilde O(f(N)) = f(N) \poly(\log f(N))}$.} Note that if in our graph (last row) we merge groups of $n$ right nodes into 1 node, both the right size and the expansion shrink by a factor of $n$, and we obtain the graph announced in the abstract.

\begin{minipage}{0.95\textwidth}
  \bigskip
  \begin{center}
    Graphs with expansion up to $K$ and left size $2^n$. 

    The first is non-explicit, the others explicit. $\alpha$ is an arbitrary positive constant.

    \smallskip
    \begin{tabular}{|l|l|l|l|}
      \hline
      left degree $D$ & right size & expansion & reference \\
      \hline
      {\color{gray} $n+2$} & \color{gray} $4K$ & \color{gray}1  & \color{gray}\cite[remark 3.6 p46]{bmvz:j:shortlist} \\
      $O(n(\tfrac 1 \eps \log K)^{1+1/\alpha})$ & $D^2 K^{1+\alpha}$ & $(1-\eps)D$ & \cite[theorem 1.3]{guv:j:extractor} \\
      $O(n(\tfrac 1 \eps \log K)^{1+1/\alpha})$ & $D K^{1+\alpha}$   & $(1-\eps)D$ & \cite[theorem 1.3]{kal-tashma:c:unbalanced} \\
      $n(\frac{1}{\eps} \log K)^{O(\log \log K)}$ & $\frac{1}{\eps}D K \poly(n)$   & $(1-\eps)D$ & \cite[theorem 18]{lu-oli-zim:c:optimalcode} \\
      $n^{5+o(1)} $ & $O(K)$ & 1 & \cite[sect 6]{teu:j:shortlists},\cite[lemma 4]{zim:c:shortlistshortproof} \\
      $\widetilde{O}(n^2)$ & $(1+o(1))nK$ & $n$ & this paper, proposition~\ref{prop:expander_quadratic_degree} \\
      \hline
    \end{tabular}
  \end{center}
  \bigskip
\end{minipage}

\noindent
For the application with short lists, the left set is the set of all strings.

\newcommand{\propositionConstantExpander}{
  For each $K$ there exists an explicit graph with left set $\{0,1\}^*$, right size~$(1+o(1)) K$, degree $\widetilde O(|x|^2)$ of left node $x$, and $1$-expansion up to~$K$. 
}

\begin{theorem}\label{th:main}
  \propositionConstantExpander
\end{theorem}

\noindent

The proof does not require section~\ref{sec:shortlist}. 
It uses a careful analysis and recombination of materials from~\cite{ats-uma-zuc:j:expanders} and~\cite{guv:j:extractor}. 
Compared to~\cite{teu:j:shortlists,zim:c:shortlistshortproof}, the reduction of degree from roughly $n^{5}$ to $n^2$ comes from 2 improvements. 

Firstly, the composition technique from \cite[theorem 3, p158]{nis-tas:j:extract} turns lossless extractors into lossy extractors of smaller degree. 
Since we only need expanders and dispersers, we optimize this and save a factor $n^2$, see section~\ref{sec:compositionLemmas}. 

Secondly, we use that the extractor in~\cite{guv:j:extractor} is already an expander. Previously, it was only used that it is a disperser, 
and an expander was obtained by merging $n$ dispersers, (see for example~\cref{lem:unite_dispersers}). This saves a factor~$n$ in the degree, but to optimize 
the right size, we still compose this extractor with $O(\log n)$ dispersers.  This is done in section~\ref{sec:constantexpander}.

\section{Application: list-computability of shortest programs}
\label{sec:shortlist}

Fix a universal optimal Turing machine. An input string is called a {\em program}. If on input $p$ the machine halts with $x$ on the tape, we say that $p$ is a \emph{program  for $x$} (or which prints $x$).

\begin{definition}
  A program $p$ that prints a string $x$ is {\em $c$-short} if no program of length less than $|p|-c$ prints~$x$. 
\end{definition}

\noindent
In~\cite{bmvz:j:shortlist} it is proven that there exists a computable function that on input $x$ prints a list of $|x|^2$ programs containing an $O(1)$-short program for~$x$. 
It is also proven that the list size is optimal up to a constant factor, more precisely, any list containing a $c$-short program has size~$\Omega(|x|^2/c^2)$. 
Unfortunately, the runtime to generate the list is very slow: doubly exponential in $|x|$. 

In~\cite{teu:j:shortlists} a {\em polynomial time} algorithm is given that prints a list containing an $O(1)$-short program. 
The variant in~\cite{zim:c:shortlistshortproof} has list size~$O(|x|^{6+\eps})$.
We obtain~$\widetilde O(|x|^3)$. 

\begin{corollary}\label{cor:lists}
  There exists a polynomial time algorithm that on input $x$ prints $\widetilde O(|x|^3)$ programs containing an $O(1)$-short one for~$x$. 
\end{corollary}

\noindent
In the remainder of this section, we explain that this is a corollary of theorem~\ref{th:main}. 
This corollary for quasi {\em cubic} lists is directly obtained from the following corollary, which gives quasi {\em quadratic} lists 
generated from any upper bound on the minimal description lenght of~$x$. 

\begin{corollary}\label{cor:lists_upperbound}
  There exists $c$ and a polynomial time algorithm that on input $x$ and $k$ prints a list of $\widetilde O(|x|^2)$ programs of length~$k$
  that contains an $x$-printing program, provided $x$ is printed by some program of length at most~$k-c$. 
\end{corollary}

\noindent
It remains to prove this corollary. 
The proof uses the following online matching game, which is studied in the companion papers~\cite{companion-DS} and~\cite{companion-Hall}.  

\medskip
\noindent
\textbf{Online matching game up to~$K$.}
Two players, called Requester and Matcher, alternate turns. 
Together they maintain a subset $M$ of edges, which is initially empty.

Requester starts. At his turn, he may remove zero or more edges from~$M$.
After this, $M$ should contain at most $K-1$ edges. 
Finally, he must select a left node~$x$. 

At her turn, Matcher may add an edge to~$M$ that is disjoint from the other edges. 
After this, node $x$~should lay on an edge of~$M$. 
If these conditions are not satisfied, then Matcher loses. 

\begin{definition}
  A graph has  {\em online matching} up to $K$ if matcher has a strategy in which he avoids losing indefinitely.  
\end{definition}

\begin{theorem}[\cite{fel-fri-pip:j:networks}]
\label{t:ffp}
  If a graph has $2$-expansion up to~$2K$, then it has online matching up to~$K$.\footnote{The proof works for infinite graphs with bounded left degree.}
\end{theorem}

\begin{proof}[Proof of~\Cref{cor:lists_upperbound}]
  Let $G_k$ be a 4-clone of the graph from~\cref{th:main} with $K = 2^{k}$. 
  This means that we take 4 copies of this graph, and merge the left nodes. 
  This graph has $4$-expansion up to $K$, so it has $2$-expansion up to~$2K$, and therefore it satisfies the condition in~\Cref{t:ffp}. Its right size is bounded by $8K$, so the right nodes can be identified with binary strings of length $k+3$. 

  The idea is to generate all strings~$x$ that are printed by a program of length less than~$k$, 
  match them to right nodes $y$ in $G_k$, and interpret such $y$ as a program for~$x$. 
  Here are the details. 

  Consider a Turing machine $V$ that on input a string $y$ of length $k+3$, constructs $G_k$, 
  evaluates all programs of length less than $k$ in parallel, 
  matches their outputs to right nodes in~$G_k$, 
  and if $y$ is matched to a string $x$, then it outputs $x$ and halts. 

  Let $U$ be the universal Turing machine that we've fixed and let $v$ be the string such that $U(vp)$ simulates $V(p)$. 

  The algorithm that satisfies the conditions of the corollary is as follows. 
  On input $k'= k+3+|v|$ and $x$, it computes the list with all $\widetilde O(|x|^2)$ neighbors of $x$ in $G_k$, each one   prepended with the fixed string~$v$. Every string in the list has length $k'$.

  This is correct, because the neighbour $y$ to which $x$ is matched will satisfy $U(vy) = x$ by construction. 
\end{proof}

\section{Two composition lemmas} \label{sec:compositionLemmas}

A function $\F : \mcX \times \mcR \rightarrow \mcY$ defines a bipartite graph with left set $\mcX$, right set $\mcY$ 
and edges $(x, \F(x,r))$ for all $x \in \mcX$ and $r \in \mcR$. 
The argument $r$ is called the  {\em seed}. 
Typically $|\mcX| \gg |\mcY| \gg |\mcR|$. 

\medskip
\noindent
We start with 2 technical lemmas, which are optimized versions of \cite[theorem 3, p158]{nis-tas:j:extract}. 
The first allows to reduce the seed of lossless expander at the cost of making them lossy. 
For this a disperser (definded below) with a small seed is used. 
The 2nd lemma obtains a small seed disperser (from a lossless extractor). 

\begin{definition}
  Function $\D : \mcX \times \mcR \rightarrow \mcY$ is a {\em $(K,\eps)$-disperser} if $|\D(X,\mcR)| \ge (1-\eps) |\mcY|$ for every $X \subseteq \mcX$ of size~$K$. 
\end{definition}

\subsection{Seed reduction for expanders}
\label{ss:seed_reduction_expanders}

The idea is to transform a part of the input to a larger seed, 
and apply the lossless expander to the other part of the input. 
Extraction of the large seed happens with a disperser, which we define soon. 
An extra expander is used for tiny input sets from which no seed can be obtained. 
The following result will be applied with $M = |\mctR|$ and $e$ proportional to this value. 

\begin{lemma}\label{lem:construction_expander}
  Assume that, for some $M \ge 2$,
  \[
  \begin{array}{rclccccl}
    \C   \!\! &  : &\!\! \{0,1\}^{<n} \!\!\!& \times & \!\!\!\mctR & \rightarrow & \mcY  & \text{has $e$-expansion up to~$\lfloor K/M \rfloor$,}
    \\\D \!\! &  : &\!\! \{0,1\}^{n}  \!\!\!& \times & \!\!\! \mcR & \rightarrow & \mctR & \text{is an $(M,\eps)$-disperser,}
    \\\C'\!\! &  : &\!\!  \{0,1\}^n   \!\!\!& \times & \!\!\! \mcR & \rightarrow & \mcY  & \text{has $\smash{\tfrac  e M}$-expansion up to~$\min(K,2M)$.} 
  \end{array}
\]
  Then the function from $\{0,1\}^n \times (\{0, \ldots, n\} \times \mcR)$ to~$\mcY$ defined by
  \[
     (x,(i,r)) \;\longmapsto\; 
     \begin{cases}
       \C(x_1 \cdots x_i, \D(x, r)) & \textnormal{if } 0 \le i < n\\
       \C'(x,r) & \textnormal{if } i = n
     \end{cases}
  \]
  has $\frac {e - \eps |\mctR|} {2M}$-expansion up to~$K$. 
\end{lemma}

\begin{remark*}
  The degree of the new graph is $(n+1)|\mcR|$. 
  It is not possible to obtain lossless expansion, even if $\C, \C'$ and $\D$ are lossless,
  because the splitting index $i$ does not contribute to the expansion 
\end{remark*}

\begin{proof}
  We need to verify that each set $S \subseteq \{0,1\}^n$ with $|S| \le K$ has a sufficient amount of neighbors. 
  If $|S| < 2M$, then already the case $i = n$ in the construction implies expansion $e/M$, which is enough for the lemma. 
  
  Assume $S \subseteq \mcX$ with $2M \le |S| \le K$. 
  Let $T$ be the set of strings with at least $M$ extensions in $S$, thus $q \in T$ if and only if $|\{qy : qy \in S\}| \ge M$. 
  Note that $T$ forms a binary tree where the root is the empty string and the children of a node $q \in T$ are its 1-bit extensions in $T$.
  By the assumption $|S| \ge M$, the tree $T$ is nonempty (because it contains the empty string). 

  \begin{claim}\label{claim:small_tree}
    $
      |T| \ge \lfloor |S|/M \rfloor. 
      $
  \end{claim}

  \begin{proof}
    Let $T^+$ be the tree containing $T$ and all children of elements in $T$. 
    Thus for a bit $b \in \{0,1\}$, we have $qb \in T^+$ if and only if $q \in T$. 
    The leaves of $T^+$ are the strings $T^+ \setminus T$. 
    These leaves define a partition of $S$ by associating 
    $q \in T^+ \setminus T$ with $\{qy : qy{\in}S\}$ of~$S$. 
    Each such set contains less than $M$ elements, because otherwise it would be in $T$ and hence not a leaf of $T^+$. 
    Therefore, $T^+$ must have more than $|S|/M$ leaves. 
    Note that a nonempty tree with $L$ leaves contains at least $L-1$ internal nodes. 
    Thus $T^+$ has more than $|S|/M - 1$ internal nodes. 
    The claim follows, since each internal node of $T^+$ is an element of~$T$. 
  \end{proof}
  
  Recall that we need to prove that $S$ has many neighbors in the constructed graph. 
  Let $T' = T$ if $|T| \le\lfloor K/M \rfloor$ or otherwise, let $T'$ be any $\lfloor K/M \rfloor$-element subset. 
  By $e$-expansion up  $\lfloor K/M \rfloor$, we have $|\C(T',\mctR)| \ge e \cdot |T'|$. 
  Note that in our construction, we do not apply all seeds in $\mctR$, since the disperser misses a small fraction. 
  For each $q \in T'$, we only use $\D(\{qy:qy{\in}S\}, \mcR) \subseteq \mctR$. 

  Consider any family of sets $R_q \subseteq \mctR$ with $|R_q| \le r$. 
  Note that
  \[
    \big| \bigcup_{q \in T'} \C(q,\mctR \setminus R_q) \big| \;\ge\; (e - r) |T'|. 
  \]
  For $q \in T'$, let $R_q$ be the set of missed seeds, thus 
  \[
    R_q = \mctR \setminus \bigcup_{y: qy {\in} S} \D(qy, \mcR). 
  \]
  By definition of a disperser, we have $|R_q| \le \eps |\mctR|$ for all $q \in T'$. 
  The lemma follows because for $r = \eps |\mctR|$,
  \[
    \big| \bigcup_{q \in T'} \C\big(q,\D(\{qy {\in} S\},\mcR)\big) \big|  \;\ge\; (e-r)|T'| \;\ge\; (e-r) \frac {|S|} {2M}, 
  \]
  where the right inequality follows from the claim and~$|S| \ge 2M$. 
\end{proof}

\subsection{A small seed disperser from a lossless extractor} 

The construction is almost the same as in subsection~\ref{ss:seed_reduction_expanders}, but a different analysis is used. 
Now, a small seed disperser extracts a seed for the lossless extractor. 
In this way, we get a disperser for {\em large} output with almost the same degree as the initial disperser with {\em small} output.
This helps with dispersers from~\cite{guv:j:extractor} because for a map with $n$-bit input and $m$-bit output, the bitsize of the seed is 
\[
  \log n + \log^2 m + O(1). 
\]
The obtained seed has bitsize $2\log n + O(\log^2 \log n)$. 
This is better than in~\cite{zim:c:shortlistshortproof}, where the bitsize exceeds $5\log n$.

\medskip
\noindent
Let $\Delta(X; \widetilde X)$ denote the statistical distance between random variables $X$ and $\widetilde X$. 

\begin{definition}
  A function $\E:\mcX \times \mcR \rightarrow \mcY$ is a {\em $(K, \eps)$-extractor} 
  if for each $S \subseteq \mcX$ of size at least~$K$, 
  \[
    \Delta(\E(U_S, U_{\mcR}), U_{\mcY}) \le \eps, 
  \]
  where $U_S$ and $U_{\mcR}$ are independent random elements of $S$ and $\mcR$. 
\end{definition}

\begin{lemma}\label{lem:construction_disperser}
  Assume that, for some $M \leq K$, $\E : \{0,1\}^{<n} \times \mctR \rightarrow \mcY$ is a $(\lfloor K/M \rfloor,\eps)$-extractor and $\D : \{0,1\}^{n} \times \mcR \rightarrow \mctR$ is an $(M,\delta)$-disperser. 
  Then the function from $\{0,1\}^n \times (\{0,\ldots, n-1\} \times \mcR)$ to~$\mcY$ defined by
  \[
     (x,(i,r)) \;\mapsto\;
       \E(x_1 \cdots x_i, \D(x, r)) 
  \]
  is a $(K,\delta + \eps)$-disperser. 
\end{lemma}

\noindent

\if01
In the proof we use the following properties of the statistical distance~$\Delta$. 
\begin{enumerate}[label=(\alph*)]
  \item 
    $\Delta(f(U,Z); f(V,Z)) \le \Delta(U,V)$ for deterministic $f$, (even if $U,V,Z$ are dependent). 
  \item 
    Let $\operatorname{supp}(X)$ denote the set of values for which $X$ has positive probability. 
    If $\Delta(U_S; X) \le \alpha$, then $|\operatorname{supp}(X) \cap S| \ge (1-\alpha)|S|$. 
\end{enumerate}
\fi

\begin{proof}
  Fix $S \subseteq \{0,1\}^n$ with $|S| = K$. We need to prove that
  \begin{equation*} 
    \Big| \bigcup_{x \in S} \bigcup_{i < n} \E(x_1 \cdots x_i, \D(x,\mcR)) \Big| \ge (1 - \eps - \delta) |\mcY|. 
  \end{equation*}
Let
  \begin{align*}
    S_p & = \{py : py \in S\} \\
    T   & = \{p : |S_p| \ge M\}.  
  \end{align*}
  It is enough to show that
\begin{equation}\tag{$*$}\label{eq:goal}
    \Big| \bigcup_{p \in T} \E(p, \D(S_p,\mcR)) \Big| \ge (1 - \eps - \delta) |\mcY|. 
  \end{equation}
  Let the set $C$ be the union in the left hand side. Let $B = \{(p,\widetilde{r}) : p \in T, \widetilde{r} \in D(S_p, \mcR)\}$. Note that $C = \bigcup_{(p,\widetilde{r}) \in B} E(p,\widetilde{r})$. 

  For every $p \in T$, $|S_p| \ge M$, and then the property of the disperser implies that $D(S_p, \mcR)$ occupies a $(1-\delta)$ fraction of $\mctR$.  In other words, $\prob [(U_T, U_{\mctR}) \in B] \ge (1-\delta)$. Then, 
  \[
  \prob [E(U_T, U_{\mctR}) \in C] \ge \prob [(U_T, U_{\mctR}) \in B] \ge (1-\delta).
\]
By claim~\ref{claim:small_tree}, $|T| \ge \lfloor K/M \rfloor$, and we can use the property of the extractor to obtain that $U_{\mcY}$ and  $E(U_T, U_{\mctR})$ are $\eps$-close in statistical distance. This implies $\prob[U_{\mcY} \in C] \ge 1-\delta - \eps$, which finishes the proof of~\cref{eq:goal}, and also of the lemma.
\end{proof}
\if01

  By definition of extractor,
  \[
    \Delta(\E(U_T, U_\mctR); U_\mcY) \le \eps. 
  \]
  We explain that if $U_\mctR$ is replaced by any random variable $V$ with $\Delta(U_\mctR, V) \le \delta$, then 
  \[
    \Delta(\E(U_T, V); U_\mcY) \le \eps + \delta. 
  \]
  Indeed, this follows by applying (a) to the assumption and by the triangle inequality. 
  Note that $V$ may be dependent on $U_T$. 

  Let $R = \D(S_{U_T}, \mcR)$ be a random set. We verify that $\Delta(U_\mctR; U_R) \le \delta$. 
  Indeed, by choice of $T$, we have $|S_{U_T}| \ge M$ with probability 1, and since $\D$ is a disperser, this implies $\smash{|R| \ge (1-\delta)|\mctR|}$. 
  We conclude that
  \[
    \Delta(\E(U_T, U_R); U_\mcY) \le \eps + \delta.  
  \]
  Property (b) above implies 
  \[
    \Big| \bigcup_{p \in T} \E(p,\D(S_p,\mcR)) \Big| \ge (1 - \eps - \delta) |\mcY|. 
  \]
  This implies \eqref{eq:goal}. The lemma is proven.
\end{proof}
\fi

\section{Construction}
\label{sec:constantexpander}

All explicit constructions rely on the extractor from~\cite[theorem 4.21]{guv:j:extractor}. 

\begin{theorem}[\cite{guv:j:extractor}]\label{th:guv}
  There exists an explicit family of extractors that contains for each $N, M$ and $\eps>0$ a $(M, \eps)$-extractor $\mcX \times \mcR \rightarrow \mcY$
  with $|\mcX| = N$, $|\mcR| = \log N (\tfrac 1 \eps \log M)^{O(\log \log M)}$ and $|\mcY| \ge \Omega(\eps^2 M|\mcR|)$. 
\end{theorem}

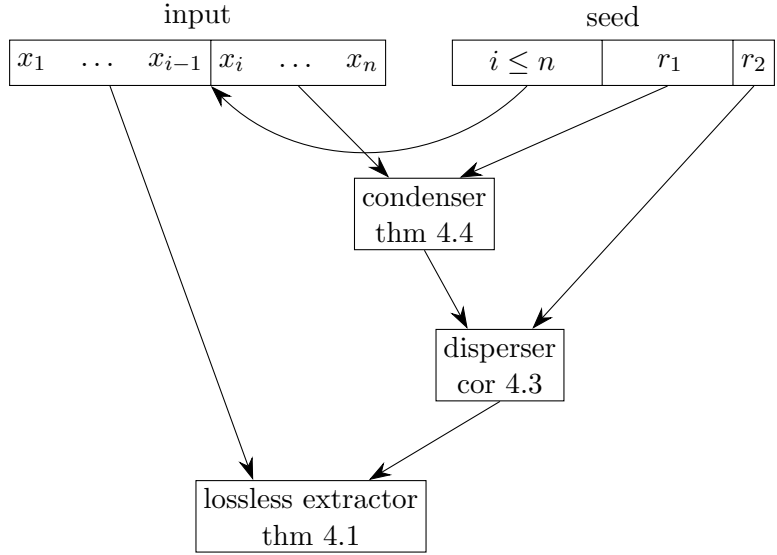
\begin{figure}[h]
  \centering
  \newcommand{\hs}{\hspace*{1em}}
  \begin{tikzpicture}[every node/.style={minimum height=6mm, inner sep=1mm}, >={Stealth[length=3mm,width=2mm]}, every text node part/.style={align=center}]
    \node[draw, rectangle split, rectangle split horizontal, rectangle split parts=2, label=above:{input}] (x) at (-3, 6) 
            {$x_1 \hs \ldots\hs  x_{i-1}$ \nodepart{two} $x_i\hs  \ldots\hs  x_n$}; 
    \node[draw, rectangle split, rectangle split horizontal, rectangle split parts=3, label=above:{seed}] (r) at (2.5, 6)  
	    {{\hs $i \le n$ \hs} \nodepart{two} \hs\; $r_1$ \;\phantom{m} \nodepart{three} $r_2$ };
    \node[draw, rectangle] (cond) at (0,4) {condenser \\thm~\ref{th:guv_condenser}};
    \node[draw, rectangle] (disp) at (1,2) {disperser \\cor~\ref{cor:main_disperser}};
    \node[draw, rectangle] (exp) at (-1.5, 0) {lossless extractor \\thm~\ref{th:guv}};
    \draw[->] (x.two south) -- ($(cond.north west)!0.25!(cond.north east)$);
    \draw[->] (r.two south) -- ($(cond.north west)!0.75!(cond.north east)$);
    \draw[->] (r.one south) to[out=-135, in=-45] (x.text split |- x.south);
    \draw[->] (r.three south) -- ($(disp.north west)!0.75!(disp.north east)$);
    \draw[->] (cond.south) -- ($(disp.north west)!0.25!(disp.north east)$);
    \draw[->] (disp.south) -- ($(exp.north west)!0.75!(exp.north east)$);
    \draw[->] (x.one south) -- ($(exp.north west)!0.25!(exp.north east)$);
  \end{tikzpicture}
  \caption{Construction of the disperser in~\cref{prop:better_disperser}. }
  \label{fig:overviewConstructionDisperser}
\end{figure}


\subsection{An explicit disperser}

\begin{proposition}\label{prop:better_disperser}
  There exists an explicit family of graphs that for each $n$ and $K$ contains 
  a $(K, 2/\log \log n)$-disperser with left size $2^n$, degree $n^2 \poly(\log n)$ and right size~$Kn$. 
\end{proposition}

The proof is outlined in figure~\ref{fig:overviewConstructionDisperser} and starts with the construction of an intermediate disperser in~\Cref{cor:main_disperser}.
This disperser (and an analog for expanders) would suffice to obtain expanders of degree $n^{2+o(1)}$ in  theorem~\ref{th:main}.
However, to decrease the degree to $n^2 \poly(\log n)$, we need to apply this technique twice, and this is more technical. Also,  for the final construction, we use a $(K, \eps)$-disperser with the right side $Kn$ instead of $K$, 
thus it ``expands'' sets of size $K$ by a factor~$(1-\eps)n$.\footnote{
  It is possible to prove a stronger version of lemma~\ref{lem:construction_disperser} in which a part of the seed is added to the output, 
  and to apply the improved disperser one more time inside lemma~\ref{lem:construction_disperser}.  
  To avoid technicality, we give a different solution. 
} These improvements are achieved subsequently in the next step.

\smallskip

Note that a $(K, \eps)$-extractor is also a $(K, \eps)$-disperser, because any variable $Y$ that is supported on a set of size less than $(1-\eps) \mcR$ has statistical distance more than $\eps$ from $U_{\mcR}$.
We apply lemma~\ref{lem:construction_disperser} with both the extractor and disperser being the graph from~\cref{th:guv}.

\begin{corollary}\label{cor:main_disperser}
  There exists an explicit family of graphs that for each $N, K$ and $\eps>0$ with $K \le N$, contains 
  a $(K, \eps)$-disperser with left size $N$, right size at least $K$ and degree 
  \[
    \tfrac 1 {\eps^2} (\log^2 N) 2^{O(\log^2 (\tfrac 1 \eps \log \log N))}. 
  \]
\end{corollary}

\begin{proof}
 Let \(n=\lceil\log_2 N\rceil\). Let $M$ be the least power of two at least
\[
n\,2^{a\log^2((\log K)/\eps)},
\]
where $a$ is chosen so that $M$ bounds the number of seeds required by Theorem 4.1 on $n$-bit inputs for every threshold $K < N$, at error $\eps$.

  \noindent
  If $K < M$, we use the complete bipartite graph with left size $N$ and right size $K$.

  \noindent
  So, we next assume $K \ge M$. We apply lemma~\ref{lem:construction_disperser} with:
  \begin{itemize}[leftmargin=*,label=--]
    \item 
      The lossless extractor from~\cref{th:guv}. 
      The size of the set of seeds is at most
      $M$.
      (We only need output size $K/M$ instead of $K$ in~\cref{lem:construction_disperser}, so this is an overestimation. 
      But it doesn't matter because of the big-O.)
    \item 
      The disperser with output size $M$ is again the extractor from~\cref{th:guv}, and its seedsize is 
      \[
	n 2^{O(\log^2 (\tfrac 1 \eps \log M)} \le n 2^{O(\log^2 (\tfrac 1 \eps \log n + \tfrac 1 \eps \log^2 \tfrac {\log K} \eps))}. 
      \]
  \end{itemize}
  This seedsize is multiplied by $n$ in~\cref{lem:construction_disperser}. 
  Using $\log K \le n$, the total degree is at most
  \[
    n^2 2^{O(\log^2 (\tfrac 1 \eps \log n)}.
  \]
  The output size is $(1-\eps) \cdot (K/M) \cdot M \cdot \eps^2 \ge \eps^2 K/2$. 
  The corollary follows by taking $\lceil \tfrac 2 {\eps^2} \rceil$ copies with disjoint right sides. The error is $2\epsilon$, but after rescaling in the argument $\epsilon \mapsto \epsilon/2$, we obtain the result.
\end{proof}

In the next step, we use a result from the same paper~\cite[theorem 4.4]{guv:j:extractor}, that is essential in the construction of theorem~\ref{th:guv}.  
Here is an equivalent statement that avoids the definition of ``condenser.'' 

\begin{theorem}[\cite{guv:j:extractor}]\label{th:guv_condenser}
  There exists a family of explicit graphs that for each $K$, $\gamma>0$ and left size $N$ contains a graph with 
  degree $D=O(\tfrac 1 \gamma \log N \log K)$, right size $(2K)^{\log D}$, 
  and which maps each set of size $K$ to a set of size $(1-\gamma)DK$.
\end{theorem}

\begin{lemma}\label{lem:disperser_smallK}
  Let $c$ be constant. 
  There exists an explicit family of graphs that for each $n$ and $K \le 2^{\log^c n}$ contains 
  a $(K, 1/\log \log n)$-disperser 
  \[
    [2^n] \times [O(n  \log^{3c+3} n)] \rightarrow [Kn]. 
  \]
\end{lemma}
\begin{proof}
  \bcom{All but 1 display in this proof had at least 1 mistake \ldots}
  Let $\C$ be a graph from theorem~\ref{th:guv_condenser} with $\gamma=1/2$.  
  We estimate the degree $D'$ and output size $\tilde n$. 
  Let $k = \log K$. Note that $k \le \log^c n$. Let $\eps = 1/\log \log n$
  \begin{align*}
    D' &= O(\tfrac 1 \gamma nk) \le O(n  \log^c n) \\
    \tilde n &= O(k \log D') \le O(\log^{c+1} n). 
  \end{align*}
  It maps a set of size $K$ to a set of size $KD'/2 \ge Kn$. 

  Consider a $(Kn, \eps)$-disperser $\D$ from corollary~\ref{cor:main_disperser} with input bitsize~$\tilde n$. The output set has size~$Kn$ 
  and the degree is
  \[
    \frac {\tilde n^2} {\eps^2}  2^{O(\log^2 (\tfrac 1 \eps \log \tilde n))} \le O(\log^{2c+2+o(1)} n). 
  \]
  The combined graph is
  \[
    (x, (r, \tilde r)) \mapsto \D(\C(x,\tilde r), r). 
  \]
  The degree of this graph is 
  \[
    O(\frac n \eps \log^c n) \cdot O(\log^{2c+2+o(1)}) \le O(n \log^{3c+3} n),
  \]
  as required. 
  This graph is a $(K, \eps)$-disperser, because the condenser maps a set of size at least $K$ 
  to a set of size $KD'/2 \ge Kn$, which in turn is mapped to~$(1-\eps)$ fraction of $\D$'s outputs. 
\end{proof}

\begin{proof}[Proof of~\cref{prop:better_disperser}.] 
  Let $\mctR = [2^{a\log^2 n}]$ where $a$ is a large enough constant such that $|\mctR|$ exceeds the degree of each extractor in~\cref{th:guv} 
  with $N = 2^n$, $M \le 2^n$ and $\eps = 1/\log \log n$. 
  Apply lemma~\ref{lem:construction_disperser} with $M = |\mctR|/n$, 
  \\-- a $(\frac K M, \frac 1 {\log \log n})$-extractor from theorem~\ref{th:guv},
  \\\phantom{x} \hfill having right size $|\mcY| \ge \Omega(\eps^2 \frac K M |\mctR|) \ge \Omega(\frac {Kn} {\log^2 \log n})$,\footnote{
    To be pedantic, $|\mctR|$ is an upper bound on the degree which could be far of and hence the requirement on $|\mcY|$ is not satisfied.  
    This is solved by taking many copies of the graph from theorem~\ref{th:guv}.
  }
  \\-- a $(\frac {|\mctR|} {n}, \frac 1 {\log \log n})$-disperser from lemma~\ref{lem:disperser_smallK} for $c=3$,  
  \\\phantom{x} \hfill having degree $O(n \log^{12} n)$.

  \medskip
  \noindent
  Thus, we obtain a $(K, 2/\log \log n)$-disperser of degree $n \cdot O(n \log^{12} n) = n^2 \poly(\log n)$. 
  Finally, the right size $|\mcY|$ is a factor $O(\log^2 \log n)$ too small, therefore we take the same number of copies, 
  which does not violate the degree requirements. 
\end{proof}

\subsection{An explicit expander}

We construct here an explicit $n$-expander up to~$K$ with right size $K \poly(n)$ (which will be reduced to $(1+o(1))Kn$ in the next section).

The proof relies on the composition lemma~\ref{lem:construction_expander}, and one component is
the following lossless expander, which is essentially the extractor from \cite[theorem 4.21]{guv:j:extractor}.
The proof that this construction also provides an expander is not too hard, 
for example it is proven in~\cite[appendix D]{bau-zim:t:univcompression} using online matching.\footnote{
  It can be shown that the graph in the proof of~\cite[theorem 4.21]{guv:j:extractor} is already an expander (with entropy loss $2\log 1/\eps + O(1)$). 
  In previous applications, a hash was added to satisfy a unique neighbor property, 
  but this is not needed here.
} 
In~\cite[theorem 18]{lu-oli-zim:c:optimalcode} the same condenser from theorem~\ref{th:guv_condenser} is applied to decrease the degree a little.

\begin{theorem}[\cite{lu-oli-zim:c:optimalcode}\footnote{The right size in~\cite[theorem 18]{lu-oli-zim:c:optimalcode} is stated in a different way, but the version used here can be easily established with a similar proof. Actually, the condition $K > \log N$ can be removed, but proving this is more complicated.} ]\label{th:loz}
  There exists an explicit family of graphs that contains for each ${\eps>0}$, $N$ and $K > \log N$, a graph with
  with left size $N$, degree $D = \log N (\tfrac 1 \eps \log K)^{O(\log \log K)}$ and right size $\frac{1}{\eps}KD \poly(\log N)$ that has expansion $(1-\eps)D$ up to~$K$. 
\end{theorem}
We also use the following simple composition lemma yielding expanders from dispersers.
\begin{lemma}\label{lem:unite_dispersers}
  For $i = 0, \ldots, k-1$, fix $(2^i, \eps)$-dispersers with the same left set and pairwise disjoint right sets of size $2^{i}$. 
  The union has $(1-\eps)$-expansion up to~$2^k-1$ and right size $2^k-1$. 
\end{lemma}

\begin{proof}
  Let $\D_i$ denote the $i$-th disperser. Consider a left set $S$ with $|S| < 2^k$. 
  To see that $S$ has many neighbors, partition $S$ according to the binary expansion of $|S|$, 
  more precisely, let $P_0 \cup \ldots \cup P_{k-1}$ be a partition of $S$ such that $|P_i| \in \{0, 2^i\}$. 
  The neighbors of $S$ contain 
  \[
    \bigcup_{i < k} \D_i(P_i), 
  \]
  which is a union of disjoint sets, since the dispersers $\D_i$ have disjoint right sides. 
  Thus, the number of neighbors is at least
  \[
     \sum_{i < k}  |\D_i(P_i)| \ge (1-\eps) \sum_{i < k} |P_i| = (1-\eps) |S|
  \]
  Therefore, the graph of the lemma is a $(1-\eps)$-expander.   
\end{proof}

\begin{corollary}\label{cor:main_expander}
  There exists an explicit family of graphs containing for each $n$ and $K$ a graph with 
  $n$-expansion up to~$K$, left size $2^n$, degree $D = n^2 \poly(\log n)$ and right size $K \poly(n)$. 
\end{corollary}
\if01
\begin{proof}\bcom{Corrected a lot of problems.}
  The proof is similar as for~\cref{prop:better_disperser}. 
  Again, let $\mctR = [2^{e\log^2 n}]$ where $e$ is a large constant such that
  $|\mctR|$ exceeds the seed size $D$ in \ref{th:loz} for $N = 2^n$, $\eps = 1/2$ and all $K \le N$. 
  
  Apply the composition lemma~\ref{lem:construction_expander} with $M = |\mctR|/n$ and the following functions. 
  \setlist{nolistsep}
  \begin{itemize}[leftmargin=*, label=--]
    \item 
      $\C$ being a graph from~\cref{th:loz} with $D/2$-expansion up to $K/M$, thus we take $\eps = 1/2$. 
      Assume that $\C$ has degree $|\mctR|$ and expansion $e = (1-\eps)|\mctR| = nM/2$, (this can be achieved by taking copies).  
    \item 
      $\D$ being the $(\frac {|\mctR|} n, \tfrac 1 4)$-disperser from lemma~\ref{lem:disperser_smallK} with $c=3$.  
      Again, $|\mcR| = n \log^{12} n$. 
    \item 
      The additional expander $\C'$ is also from~\cref{th:loz} with expansion up to $2M$. 
  \end{itemize}
  We verify that the seed size of $\C'$ is at most $n \poly(\log n)$. 
  In the corresponding application of~\ref{th:loz} we select  $K$ to be
  \[
    2M = 2|\mctR|/n = 2^{O(\log^2 \log n)}.
  \]
  Hence, the degree in~\cref{th:loz} is 
  \[
    |\mcR| \le n 2^{O(\log^2 \log M)} \le n 2^{O(\log^2 \log \log n)} \le O(n \log n). 
  \]
  The expansion of $\C'$ is at least $(1/2) \times$ degree, which is $\omega(n)$, and this is at least $e/M = (1-\eps)n$. 

  \Cref{lem:construction_expander} now ensures expansion
  \[
    \frac {e - \tfrac 1 4 | \mctR|}{2M} \ge \frac {(\tfrac 1 2 - \tfrac 1 4) |\mctR|}{2M} = (\tfrac 1 8 - o(1))n. 
  \]
  Finally, 9 copies of the graph increases the expansion to at least~$n$. 
\end{proof}
\fi
\begin{proof}
The proof is similar to the proof of Proposition 4.2. Again, let
$\widetilde R=[2^{a\log^2 n}]$,
where $a$ is a sufficiently large constant such that $|\widetilde R|$ exceeds the degree in
Theorem 4.6 for $N=2^n$, $\eps=1/4$, and all $K\le N$. Put
$M=\frac{|\widetilde R|}{n}$.

First we handle the case when \(\lfloor K/M\rfloor\le n\). Then \(K<(n+1)M=2^{O(\log^2 n)}\). Taking the union of the dispersers from~\Cref{lem:disperser_smallK} with thresholds \(2^i\le K\), and then two copies with disjoint right sides, gives, by the argument of~\Cref{lem:unite_dispersers}, an \(n\)-expander up to \(K\), of degree \(n\operatorname{poly}(\log n)\) and right size \(O(nK)\).

We may therefore assume that \(\lfloor K/M\rfloor>n\). Now we apply Lemma 3.1 with the following functions.
\begin{itemize}[leftmargin=*,label=--]
\item
$C$ is a graph from Theorem 4.6 with $3D/4$-expansion up to
$\left\lfloor \frac{K}{M}\right\rfloor$,
taking $\varepsilon=1/4$. By taking copies and padding seeds, if necessary, we may assume that
$C$ has degree $|\widetilde R|$. Hence we may take its expansion parameter $e$ to be
$\frac{|\widetilde R|}{2}=\frac{nM}{2}$.

\item
$D$ is the
$\left(M,\frac14\right)$-disperser from Lemma 4.5 with $c=3$. Its degree is
$n\poly(\log n)$.

\item
It remains to construct the additional expander $C'$, which must have $\frac eM=\frac n2$-expansion up to $\min(K,2M)$.

For every power of two $2^i\le \min (K,2M)$, take a
$\left(2^i,\frac1{\log\log n}\right)$-disperser $D_i$ from Lemma 4.5 whose right side has size $n2^i$, and make the right sides
pairwise disjoint. Let $C'$ be their union. By~\Cref{lem:unite_dispersers}, $C'$ has $\left(1-\frac1{\log\log n}\right)n$
-expansion up to $\min(K,2M)$, which is even more than what we need.

Since
$M=\frac{2^{a\log^2 n}}{n}=2^{O(\log^2 n)}$,
Lemma 4.5 is applicable with $c=3$. Moreover, there are only
$O(\log M)=O(\log^2 n)$
graphs $D_i$, and each has degree $n\poly(\log n)$. Therefore $C'$ has degree
$n\poly(\log n)$.

The right side of $C'$ has size
$\sum_{2^i\le \min(K,2M)} n2^i=O(nK)$.

After padding the seed sets of $D$ and $C'$, if necessary, they may be taken to have the same
seed set $R$, with $|R|=n\poly(\log n)$.
Likewise, by enlarging the common right set if necessary, we may regard $C$ and $C'$ as mapping
to the same set $Y$. The resulting right side still has size $K\operatorname{poly}(n)$.
\end{itemize}

Lemma 3.1 now gives a graph of degree
\[
(n+1)|R|=n^2\operatorname{poly}(\log n)
\]
and expansion
\[
\frac{e-\frac14|\widetilde R|}{2M}
=
\frac{\frac12|\widetilde R|-\frac14|\widetilde R|}{2M}
=
\frac n8.
\]
Finally, taking $9$ copies increases the expansion to at least $n$, without changing the required
asymptotic bounds on the degree or the right size.
\end{proof}

\subsection{Reducing the right size}\label{s:nexpander}

We have seen in~\Cref{lem:unite_dispersers} that one can obtain an $(1-\eps)$-expander up to $2^k-1$ with almost optimal right size $2^k-1$ by merging $k$ dispersers with geometrically increasing right sizes.   Unfortunately, this construction increases the degree by a factor~$k$. 
We reduce this factor to $O(\log n)$ by merging fewer graphs, which is obtained by replacing the dispersers with small left size  $\le 2^{k - O(\log n)}$ with a single expander. 

\begin{proposition}\label{prop:expander_quadratic_degree}
  There exists an explicit family of graphs containing for each $n$ and~$K'$ 
  a graph with $n$-expansion up to~$K'$, left size $2^n$, degree $n^2 \poly(\log n)$ and right size $(1+O(1/\log \log n))nK'$. 
\end{proposition}

\if01
\bcom{OK, many changes. Yes, I completely missed the issue, and the idea from chatgpt is really simple. I've improved the writing.}

\begin{proof}
  If $K' < n$, then the proposition is trivially satisfied by a fully connected bipartite graph. 
  Hence, assume $K' \ge n$. 

  Let $\eps = 2/\log \log n$ and let $\E$ be the graph in corollary~\ref{cor:main_expander} with right size $\lceil \eps nK'\rceil$. 
  This graph has $n$-expansion up to $\delta K'$ with $\delta = \eps/\poly(n)$. 

  For any positive integer~$i$, obtain from~\cref{prop:better_disperser}, 
  a $(\lceil K'2^{-i}, \eps)$-disperser $\D_i$ with right size $n \lceil K'2^{-i}/(1-\eps) \rceil$. 
  Note that the right size is slightly larger, so one needs to merge 2 dispersers for suitable values of~$K$. 

  The graph that satisfies the conditions of the proposition is the union of $\E$ and the graphs
  \[
    \D_{1}, \D_{2}, \ldots, \D_{\lceil \log \tfrac 1 \delta\rceil}, 
  \]
  where all right sides are disjoint and left sides are unified. 
  The right side of this union has size bounded by 
  \[
    \le \eps n K' + \frac {n(K' + \lceil \log \tfrac 1 \delta \rceil)}{1-\eps} + 1 \le (1+O(\eps)) nK', 
  \]
  where we used $\log \tfrac 1 \delta \le O(\log n) \le \eps n \le \eps K'$. 
  The degree of each component in the union is $n^2 \poly(\log n)$, and there are $O(\log \frac 1 \delta) \le O(\log n)$ components. 
  Thus, the right size and the degree are as claimed.

  The expansion is proven in a similar way as in~\cref{lem:unite_dispersers_allsizes}: the set $S$ is partitioned in parts $P_1, \ldots, P_{\log \tfrac 1 \delta}$ and $P_{\E}$, 
  corresponding to graphs $\D_i$ and graph $\E$. The part $P_{\E}$ contains all elements not taken by the other parts, and there are at most $\delta|S|$ of them. 
\end{proof}
\fi
\begin{proof}
  Let $\eps = 2/\log \log n$. 
  Let $\E$ be the graph in corollary~\ref{cor:main_expander} with right size $\eps K'$. 
  Thus, this graph has $n$-expansion up to $\eps K'/p(n)$, where $p(n)$ is a polynomial. 
  Let $d_n = \lceil \log (\tfrac 1 \eps p(n)) \rceil$. 

For any positive integer $A$, we obtain from~\cref{prop:better_disperser},  an $(A,\eps)$ disperser  in which every left set of size $A$ has at least $nA$ neighbors.  This is done by taking the left set to be the set of strings $n' = \lceil\frac{n}{1-\eps} \rceil$, and restrict it to $n$-bit strings padded to length $n'$. Call this graph $D_A$. It has right size $An'$ and degree $n'^2 \poly(\log n') = n^2 \poly(\log n)$. 
  
  We now construct our graph as a union of $E$ and graphs of the type $D_A$. We first consider the simpler case when $K'$ is a power of $2$. Let $k = \log K'$. We will assume that $k > d_n$, because the case of a smaller $k$ can be handled in a similar and actually easier way.
  The graph that satisfies the conditions of the proposition is the union of $\E$ and the graphs
  \[
    \D_{2^{k-1}}, \D_{2^{k-2}}, \ldots, \D_{2^{k-d_n}}, 
  \]
  where all right sides are disjoint and left sides are unified. The right side of this union has size bounded by $ \eps K' + \lceil n /(1-\eps) \rceil \sum_{i=k-d_n}^{k-1} 2^i   = n K'\ (1+ O(\eps))$. The degree of each component in the union is $n^2 \poly(\log n)$, and there are $d_n + 1 = \poly(\log n)$ components. Thus, the right size and the degree are as claimed.

  Expansion is guaranteed in a similar way as in lemma~\ref{lem:unite_dispersers} 
  by partitioning $S$ in $d_n + 1$ parts corresponding to the above graphs. 
  More precisely, the parts $P_i$ for $i = k-1, \ldots, k-d_n$ have either cardinality $0$ or $2^i$, 
  and the at most $2^{k-d_n}$ remaining elements are assigned to~$P_{\E}$. 
  Each part $P_i$ has $n$-expansion in $D_i$, for $i=2^{k-1}, 2^{k-2}, \ldots, 2^{k- d_n}$.
  By choice of $d_n$, part $P_{\E}$ has $n$-expansion in~$\E$. 
  Thus, $S$ has $n$-expansion in the union. 

  For the general $K'$ (not a power of $2$), we  adjust the thresholds $2^{k-d_n}, \ldots, 2^{k-1}$ to be $A_{d_n}, \ldots, A_1$, defined as follows. For $1 \le j \le d_n$, we take $A_j = \lfloor \frac{K'}{2^{j-1}}\rfloor  - \lfloor \frac{K'}{2^{j}} \rfloor$. 
  The key fact is that
every $s\le K'$ can be written as 
\[s=\sum_{j=1}^{d_n}b_jA_j+r,\quad b_j\in\{0,1\},\quad
0\le r\le\Big\lfloor K'/2^{d_n}\Big\rfloor,
\]
and therefore $S$ can be partitioned similarly to the easy case.  
  
  By telescoping, $\sum_{j=1}^{d_n} A_j \le K'$, and thus the right side of the union has size bounded by $\eps K' + \lceil n /(1-\eps) \rceil \sum_{j=d_n}^{1} A_j   = n K'\ (1+ O(\eps))$. 
  
  The (tedious) details are in~\cref{s:gencase}. They were obtained with chatGPT-6.
\end{proof}

\subsection{Extending the left set to strings of arbitrary length}

Finally, we prove theorem~\ref{th:main}. 

\begin{theorem*}[Restated]
  \propositionConstantExpander
\end{theorem*}

\begin{proof} Let $k = \log K$. Note that $k$ may not be an integer. 
  \newcommand{\Gsmall}{G_{\textnormal{small}}}
  \newcommand{\Gmid}{G_{\textnormal{mid}}}
  \newcommand{\Glarge}{G_{\textnormal{large}}}
  \begin{itemize}[leftmargin=*, label=--]
     \item 
       Let $\Gsmall$ be the graph with left set $A=\{x : |x| \le \tfrac k 2\}$, right set equal to a copy of $A$,
       and each left node $x$ connected to its own unique neighbor on the right side.
       The right size of $\Gsmall$ is bounded by $2 \sqrt{K}$ and each left node has degree~$1$.
    
    \item 
      Let $\Glarge$ be the graph whose left nodes are all strings of length at least $\sqrt{K}$ 
      whose right set is $[K]$ and in which each left node is connected to all right nodes. 
    
    \item 
      Next, we define a sequence of graphs  to  cover the remaining nodes, having medium lengths. Let $q = 2^{\lfloor \log k \rfloor}$ and $L= \lceil \frac{K}{q} \rceil$.
      For $j \in \{\lfloor \log k \rfloor, \ldots, \lceil \frac{k}{2} \rceil \}$, let $G_j$ be a graph with left size $2^{2^j}$, 
      right size $(1+o(1))K$, and $q$-expansion up to~$L$. 
      Such a graph is obtained from \cref{prop:expander_quadratic_degree} by merging right nodes. 
      Here are the details. Start with a graph from the proposition in which

      \begin{itemize}[leftmargin=*, label=$\bullet$]
	\item the left side consists of all strings $x$ with $2^{j-1} \le |x| < 2^j$, 
	 \newline thus for $n_j = 2^j$, its left size is at most~$2^{n_j}$, 

        \item has $n_j$-expansion up to $L$, 

        \item has right size $(1+o(1))n_j L$,

        \item has left degree $(n_j)^2 \poly(\log n_j)$, 
	 \newline which is $|x|^2 \poly(\log (|x|)$ for every left node $x$, since $n_j \le 2|x|$. 
      \end{itemize}

      Next, for $T= n_j/q$ (which is an integer), merge groups of $T$ right nodes into 1 node. 
      This shrinks both the expansion and the right size by a factor of $T$  (actually, the expansion may shrink less, but we consider the most pessimistic case).   
      So, the expansion becomes at least $n_j/T = q$ and the right size becomes $[(1+o(1))n_j L ]/T$ = $(1+o(1)) K$ (taking into account that $q= O(\log K)$).
  \end{itemize}

  \medskip
  \noindent
  Merge the right sides of the graph $\Glarge$ and all the graphs $G_j$ for 
  \[
     j = \lfloor \log k \rfloor, \ldots, \lceil \tfrac k 2 \rceil.
   \]
  The final graph is the disjoint union of the above graph with $\Gsmall$.

  By construction, the final graph has right size bounded by $(1+o(1))K + 2 \sqrt{K} = (1+o(1))K$. 
  The degree of each  left node $x$ is $O(|x|^2 \poly(\log |x|))$. 

  It remains to verify that the final graph has $1$-expansion up to~$K$. 
  Let $S$ be a set with $|S| \le K$. 
  If $S$ contains a string of length at least $\sqrt{K}$, then $S$ immediately has $K$ neighbors and we are done.  
  Also, nodes of length at most $\tfrac k 2$ have their own unique neighbor in $\Gsmall$ and so (because of the disjoint union) also in the final graph. 
  Therefore, in the analysis, they can be discarded from~$S$.

  It remains to analyze the case $S \subseteq \{x : \tfrac k 2 < |x| < \sqrt{K}\}$. 
  Note that there are less than $\lfloor \tfrac k 2\rfloor < q$ graphs $G_j$ in the construction. 
  Hence, there exists a $G_j$ whose left set contains at least $|S|/q$ elements of~$S$. 
  Let $S_j$ be this left set. 

  Finally we explain that $q$-expansion implies that $S_j$ has at least $|S|$ neighbors in $G_j$. 
  Indeed, if $|S_j| \le L$, then by $q$ expansion up to $L$, the number of neighbors is $q|S_j| \ge |S|$. 
  Otherwise, select any subset $T$ of $S_j$ of size $L$. Again by $q$-expansion, $T$ has at least $qL \ge K$ neighbors, 
  while by assumption $|S| \le K$. 
  In conclusion, $S_j$ has at least $|S|$ neighbors, and the combined graph has $1$-expansion up to $K$. 
  The theorem is proven.
\end{proof}

\if01
\newpage

  \marius{[revised proof -using the repair from chatGPT]}
\begin{proof}
\newcommand{\Gsmall}{G_{\textnormal{small}}}
  \newcommand{\Glarge}{G_{\textnormal{large}}}
  \newcommand{\Lsmall}{L_{\textnormal{small}}}
  \newcommand{\Llarge}{L_{\textnormal{large}}}
    Let $L = \lceil \log K \rceil$ and $H = \lceil \frac{K}{L} \rceil$. We have $K \le L H \le (1+o(1)) K$.

    We define 3 sets of strings:

    SMALL: the set of strings of length at most $\lfloor L/2  \rfloor$,

    MIDSIZE: the set of strings of length  $> \lfloor L/2  \rfloor$ and $\le \sqrt{K}$, and

    LARGE: the set of strings of length  $> \sqrt{K}$.

    \begin{itemize}[leftmargin=*, label=--]
     \item 
       Let $\Gsmall$ be the graph with left set equal to SMALL, right set equal to a copy of SMALL,
       and each left node $x$ connected to its own unique neighbor on the right side.
       The right size of $\Gsmall$ is bounded by $O( \sqrt{K})$ and each left node has degree~$1$.
    
    \item 
      Let $\Glarge$ be the graph whose left nodes is the set LARGE, 
      whose right set is $[K]$ and in which each left node is connected to all right nodes. 
    cc
    \item  For every $j$ such that $\lfloor L/2 \rfloor < 2^j \le 2 \sqrt{K}$, we define  the subset of MIDSIZE,
    \[
    X_j = \{x \in \textrm{MIDSIZE} \mid 2^{j-1} \le |x| < 2^j\}.
    \]
Let $n_j = L \lceil \frac{2^j}{L} \rceil$. Proposition~\ref{prop:expander_quadratic_degree} provides a graph in which left nodes have length $n_j$, has $n_j$-expansion up to $H$, has right size $(1+o(1)) n_j H$, and left degree $n_j^2 \poly \log n_j$. Restrict its left side to the embedded class $X_j$. Now merge its right vertices into groups of size at most
\[
T_j=n_j/L,
\]
which is an integer. This reduces both the expansion and the right size by a factor $T_j$. Call the resulting graph $G_j$. This graph has $L$-expansion up to $H$, right size $(1+o(1)) LH$, and all the left nodes $x$ have degree $|x|^2 \poly \log |x|$, because $n_j$ is $O(|x|)$.

    \end{itemize}

    Merge the right sides of the graph $\Glarge$ and all the graphs $G_j$.
  The final graph is the disjoint union of the above graph with $\Gsmall$.

  By construction, the final graph has right size bounded by $(1+o(1))K + O(\sqrt{K}) = (1+o(1))K$. 
  The degree of each  left node $x$ is $O(|x|^2 \poly(\log |x|))$. 

  It remains to verify that the final graph has $1$-expansion up to~$K$. 
  Let $S$ be a set with $|S| \le K$. 
  If $S$ contains a string in the set LARGE, then $S$ immediately has $K$ neighbors and we are done.  
  Also, nodes in the set SMALL have their own unique neighbor in $\Gsmall$ and so (because of the disjoint union) also in the final graph. 
  Therefore, in the analysis, they can be discarded from~$S$.

  It remains to analyze the case $S \subseteq \textrm{MIDSIZE}$.
Note that there are less than $k = \lfloor \tfrac 1 2 \log K\rfloor$ subsets $X_j$ in the construction. 

There are two cases.
Case 1: Some subset $X_j$ contains at least \(H\) elements of $S$.
Choose an \(H\)-element subset \(A\) in that class. Then
\[
|N(S)|
\ge |N_{G_j}(A)|
\ge LH
\ge K
\ge |S|.
\]
Case 2: Every subset contains fewer than \(H\) elements.
By the pigeon principle, a largest intersection $X_j \cap S$ satisfies
\[
|X_j \cap S|\ge |S|/k.
\]
Its size is less than $H$ (because we in Case 2). So it lies within the expansion range, and hence
\[
|N(S)|
\ge |N_{G_j}(X_j \cap S)|
\ge L|X_j \cap S|
\ge \frac Lk|S|
\ge |S|.
\]
  
  The theorem is proven.
  
\end{proof}
\fi
  \section*{AI Disclosure}
ChatGPT (the GPT-6 model family) has proofread  the paper in a late stage of its development, namely version 1 posted on arxiv with a different title, which is entirely done by the authors. It found a few minor technical issues, which we fixed.

\newcommand{\etalchar}[1]{$^{#1}$}

\appendix

\section{Proof of the general case of~\Cref{prop:expander_quadratic_degree}}
\label{s:gencase}

For completeness, we give the full proof for the case when \(K'\) is not necessarily a power of \(2\). This proof was delivered by chatGPT-6, at our request.

Put
\[
d=d_n=\left\lceil \log\!\left(\frac{p(n)}{\varepsilon}\right)\right\rceil
\]
and, for \(j=0,\ldots,d\), let
\[
Q_j=\left\lfloor\frac{K'}{2^j}\right\rfloor .
\]
Thus \(Q_0=K'\). For \(j=1,\ldots,d\), define
\[
A_j=Q_{j-1}-Q_j
   =\left\lfloor\frac{K'}{2^{j-1}}\right\rfloor
      -\left\lfloor\frac{K'}{2^j}\right\rfloor .
\]
For every \(j\) with \(A_j>0\), take a separate copy of the graph \(D_{A_j}\), and make all these right sides, as well as the right side of \(E\), pairwise disjoint. Let \(G\) be their union.

First note that
\[
\sum_{j=1}^{d}A_j
   =Q_0-Q_d
   =K'-\left\lfloor\frac{K'}{2^d}\right\rfloor
   \le K'.
\]
Consequently, using the fact that \(D_A\) has right size
\[
A n',\qquad
n'=\left\lceil\frac{n}{1-\varepsilon}\right\rceil ,
\]
the right side of \(G\) has size at most
\[
\varepsilon K'
   +n'\sum_{j=1}^{d}A_j
\le
\varepsilon K'
   +\left\lceil\frac{n}{1-\varepsilon}\right\rceil K'
=
nK'(1+O(\varepsilon)).
\]
Here we use \(\varepsilon=2/\log\log n\). Also, each component has degree
\(n^2\operatorname{poly}(\log n)\), and there are only
\[
d+1=O(\log n)
\]
components. Hence the degree of \(G\) is still
\[
n^2\operatorname{poly}(\log n).
\]

It remains to verify \(n\)-expansion up to \(K'\).

Let \(S\) be an arbitrary left set with
\[
|S|\le K'.
\]
We show that \(S\) can be partitioned as
\[
S=P_1\mathbin{\dot\cup}\cdots\mathbin{\dot\cup}P_d
   \mathbin{\dot\cup}P_E
\]
so that, for every \(j\),
\[
|P_j|\in\{0,A_j\},
\]
and
\[
|P_E|\le Q_d
=\left\lfloor\frac{K'}{2^d}\right\rfloor.
\]

We construct the partition successively. Initially let \(R_0=S\). Inductively suppose that
\[
|R_{j-1}|\le Q_{j-1}.
\]
Observe that
\[
Q_j=\left\lfloor\frac{Q_{j-1}}2\right\rfloor
\]
and therefore
\[
A_j=Q_{j-1}-Q_j
   =\left\lceil\frac{Q_{j-1}}2\right\rceil .
\]

If
\[
|R_{j-1}|\le Q_j,
\]
put \(P_j=\varnothing\) and \(R_j=R_{j-1}\).

Otherwise \(|R_{j-1}|>Q_j\). Since \(|R_{j-1}|\) is an integer,
\[
|R_{j-1}|
>
\left\lfloor\frac{Q_{j-1}}2\right\rfloor
\]
implies
\[
|R_{j-1}|
\ge
\left\lceil\frac{Q_{j-1}}2\right\rceil
=A_j.
\]
Hence we can choose an arbitrary subset
\[
P_j\subseteq R_{j-1},
\qquad |P_j|=A_j,
\]
and let
\[
R_j=R_{j-1}\setminus P_j.
\]
Since \(|R_{j-1}|\le Q_{j-1}\),
\[
|R_j|
\le Q_{j-1}-A_j
=Q_j.
\]
Thus in either case the invariant \(|R_j|\le Q_j\) holds. After \(d\) steps, put
\[
P_E=R_d.
\]
Then
\[
|P_E|\le Q_d.
\]

By the choice of \(d\),
\[
2^d\ge\frac{p(n)}{\varepsilon},
\]
and hence
\[
|P_E|
\le Q_d
\le\frac{K'}{2^d}
\le\frac{\varepsilon K'}{p(n)}.
\]
The graph \(E\) has \(n\)-expansion up to
\(\varepsilon K'/p(n)\), so
\[
|N_E(P_E)|\ge n|P_E|.
\]

For every \(j\) with \(P_j\ne\varnothing\), we have
\[
|P_j|=A_j.
\]
The graph \(D_{A_j}\) is an \((A_j,\varepsilon)\)-disperser with right side of size \(A_j n'\). Hence
\[
|N_{D_{A_j}}(P_j)|
   \ge (1-\varepsilon)A_j n'
   \ge nA_j
   =n|P_j|.
\]
This is precisely the property of the graphs \(D_A\) established earlier in the proof. 

Finally, all the component graphs have disjoint right sides. Therefore
\[
\begin{aligned}
|N_G(S)|
&\ge
|N_E(P_E)|
+\sum_{j=1}^{d}|N_{D_{A_j}}(P_j)|\\
&\ge
n|P_E|+\sum_{j=1}^{d}n|P_j|\\
&=n|S|.
\end{aligned}
\]
Thus \(G\) has \(n\)-expansion up to \(K'\).

Notice that this argument also covers the case \(K'<2^d\). In that case \(Q_d=0\), so \(P_E=\varnothing\); any indices for which \(A_j=0\) are simply omitted. Hence no separate small-\(K'\) argument is actually necessary.

\if01
The key fact missing from the draft is therefore
\[
\boxed{\text{every }s\le K'\text{ can be written as }
s=\sum_{j=1}^{d}b_jA_j+r,\quad b_j\in\{0,1\},\quad
0\le r\le\Big\lfloor K'/2^d\Big\rfloor.}
\]
The greedy construction above proves exactly this and is the general-\(K'\) substitute for the ordinary binary expansion used in the power-of-two case.
\fi
\end{document}